\documentclass[12pt]{amsart}
\usepackage{epsfig}
\usepackage{enumerate}
\usepackage{graphicx}
\usepackage{amsmath}
\usepackage{amssymb}
\usepackage{latexsym}
\usepackage{amsfonts}

\newtheorem{theorem}{Theorem}[section]
\newtheorem{proposition}[theorem]{ Proposition}
\newtheorem{lemma}[theorem]{Lemma}
\newtheorem{corollary}[theorem]{Corollary}

\newtheorem{definition}[theorem]{Definition}

\newcommand{\Z}{\mathbb{Z}}
\newcommand{\Q}{\mathbb{Q}}
\newcommand{\R}{\mathbb{R}}
\newcommand{\C}{\mathbb{C}}

\renewcommand{\a}{\alpha}
\renewcommand{\b}{\beta}
\renewcommand{\l}{\lambda}
\newcommand{\p}{\varphi}
\newcommand{\g}{\gamma}
\newcommand{\dt}{\delta}

\newcommand{\e}{\varepsilon}

\newcommand{\cL}{\mathcal{L}}

\newcommand{\cQ}{\mathcal{Q}}

\newcommand{\cR}{\mathcal{R}}
\renewcommand{\S}{\Sigma}
\newcommand{\lra}{\longrightarrow}

\newcommand{\w}{\wedge}
\renewcommand{\d}{\partial}

\title{Resurgent Deformation Quantisation}

\author{Mauricio Garay, Axel de Goursac and Duco van Straten}
\begin{document}

\maketitle

\begin{abstract}
We construct a version of the complex Heisenberg algebra based on the
idea of endless analytic continuation. The algebra would be large enough to capture quantum effects that escape ordinary formal deformation quantisation.
\end{abstract}

\section*{Introduction}
In $1925$, {\sc Born} and {\sc Jordan} introduced the non-commutative
algebra of formal power series in the variables $q,p,h$ subject to
the relation
\[ pq-qp=\frac{h}{2\pi i}\]
to explain the calculation of {\sc Heisenberg} for the spectrum of the anharmonic oscillator~\cite{Born_Jordan}.

Since then quantum mechanics  has been considered  as a deformation of classical mechanics, with $h$ as a small parameter.
However, it is also well-known that many quantities of interest
are not holomorphic in $h$ near $h=0$; the wave associated by {\sc De Broglie} to
a free particle with momentum $p =h k/2\pi$
\[ e^{2\pi i p q/h}=e^{i k q},\]
being a fundamental example in case. Further examples related to the Hamiltonian 
\[ H=\frac{1}{2m}p^2 +V(q)\]
of a particle of mass $m$ in a potential $V(q)$ are numerous: tunnel amplitudes,
in the first order of the WKB-approximation, like 
\[ e^{-\frac{2\pi}{h} \int \sqrt{2m(V(q)-E)}dq},\;\]
or the exponentially small separation between the first and the second eigenvalue of the quartic oscillator with potential $V(q)=q^4-\beta q^2$.

These phenomena lead to the fact that most series in $h$ appearing in perturbation theory are divergent and have an asymptotical meaning at best, a point of view already advocated by {\sc Birkhoff} \cite{birkhoff1933quantum}. The traditional approach to deal with such quantities is to use classical
Hilbert space analysis on the Schr\"o\-dinger equation or use
semi-classical or more general micro-local analysis~\cite{Sjostrand,Reed_Simon,zworski2012semiclassical}.

Deformation quantisation initially ignored exponentially small quantities; like in formal quantum mechanics, series in $h$ had only a formal meaning~\cite{Flato}. Nevertheless, in the late eighties,  {\sc Rieffel} constructed examples of non-formal deformation quantisations  in the real differentiable context, and since there has been several works in this direction~\cite{rieffel1989deformation}~(see also \cite{bieliavsky2011deformation}).

Parallel to real analysis, one may study these divergent expansions from the complex geometric viewpoint. It was indeed realised early (or sometimes simply conjectured)  that many of them have a property of {\em endless analytic continuation}, when expressed in a Borel transformed variable $\xi$. This led {\sc Voros} and {\sc Zinn-Justin} to exact quantisation formul\ae\ which were later explained,
by {\sc Delabaere} and {\sc Pham}, as resurgence properties of the complex WKB expansions~\cite{delabaere1997unfolding,Eremenko_Gabrielov_quartic,Voros,Zinn_Justin_nuclear,zinn2004multi}. However, this approach which gathers Voros-Zinn-Justin conjectures and resurgence analysis is for the moment still conjectural.

  The purpose of this paper is to define a {\em resurgent Heisenberg algebra} $\cQ^A$ or more precisely an algebra of resurgent operators with algebraic singularities. We hope this algebra will be rich enough to capture quantum effects beyond perturbation theory and lead to a better understanding of the complex WKB method and exact quantisation conjectures. However, for the moment, we observe that the dual star-algebra defined in this paper obeys {\sc \'Ecalle}'s philosophy that although complicated transcendental function may appear, the description of their singularities is simple and can be made explicit. For instance, we will see that Laplace transforms of hypergeometric functions appear naturally as products of algebraic functions.

\section{Heisenberg algebras}

In this section we introduce various versions of the Heisenberg algebra. As $h$, the imaginary unit $i$ and factors $2\pi$ appear in many formulas, we will set
\[ t:=\frac{h}{2\pi i}~.\]
On the polynomial ring $\C[t,q,p]$, 
we consider the (non-commutative associative) normal product $\star$ given by
\[p \star q=qp+t,\;\; q \star p=qp,\]
and furthermore
\[p \star p^{n-1}=p^n,\;\;q \star q^{n-1}=q^n,\;\;t \star p=p \star t =tp,\;\;t \star q=q \star t=tq,\]
where on the right hand side we use the ordinary product of polynomials.
The resulting algebra with product $\star$ is known as the  {\em Heisenberg algebra} and will be denoted by $\cQ$. The mapping 
$q \mapsto q,\;\;p \mapsto t\frac{d}{dq},$
identifies $\cQ$ with the {\em Weyl-algebra} of $t$-differential operators
$\cQ \cong \C<t,q,t\frac{d}{dq}>$.

When we write elements $f,g \in \cQ$ as $f=\sum_{n\ge 0} f_n t^n$, $g=\sum_{n \ge 0} g_n t^n$,
with coefficients $f_n, g_m \in \C[q,p]$, we can expand the $\star$-product of
$f$ and $g$ as
\[ f \star g= \sum_{l \ge 0} h_l t^l.\]
\begin{proposition}[\cite{Moyal}]\label{P::Moyal} The coefficient $h_l$ in the expansion
\[ f \star g= \sum_{l \ge 0} h_l t^l\]is given by the formula
\begin{equation}
h_l=\sum_{n+m+k=l} \frac{1}{k!} \frac{\partial^k f_n(q,p)}{\partial p^k}\frac{\partial^k g_m(q,p)}{\partial q^k}. \label{eq-starformel}
\end{equation}
where $f=\sum_{n\ge 0} f_n t^n$, $g=\sum_{n \ge 0} g_n t^n$.
\end{proposition}

As these expressions make sense for formal power series, one can use this 
formula to obtain a $\star$-product on $\C[[t,q,p]]$. The resulting algebra  
we call the {\em formal Heisenberg algebra} and denote it by $\widehat{\cQ}$.
Clearly $\cQ \subset \widehat{\cQ}$.

There are various interesting algebras between $\cQ$ and $\widehat{\cQ}$,
for example the algebras $\C[[t]][q,p]$ and $\C[q,p][[t]]$,
that appear naturally in constructions that proceed order-by-order in $t$
or $h$. But in this paper we will be interested in quite different sub-algebras
of $\widehat{\cQ}$ that are characterised by analytic properties and analytic continuation.

\subsection*{There is no $\star$-algebra of analytic operators.}
It is a fundamental fact that it is not possible to define a $\star$-algebra of analytic operators. Even for  meromorphic functions, the $\star$-product leads in general to divergent series and is therefore ambiguous. We can observe this fact by explicit computation. Let us denote by
\begin{equation*}
E(t):=\sum_{n=0}^\infty n! t^n
\end{equation*}
the power series considered by {\sc Euler}~\cite{Euler_divergent}.
\begin{proposition}
\label{P::Euler}
The star-product of $\frac{1}{1-p} $ and $ \frac{1}{1-q}$ is a divergent series given by the formula
$$ \frac{1}{1-p} \star \frac{1}{1-q}=\frac{1}{(1-p)(1-q)}E\left(\frac{t}{(1-p)(1-q)}\right). $$
\end{proposition}
\begin{proof}
We have to compute $\sum_{n,m\ge 0} p^n \star q^m$. From the formula \eqref{eq-starformel} of the $\star$-product we find
$$
\begin{array}{rcl}
p^n \star q^m&=&\sum_{k \ge 0} \frac{1}{k!}\partial^k_p p^n \partial_q^kq^m t^k=\sum_{k \geq 0}k!\begin{pmatrix}n\\ k  \end{pmatrix} \begin{pmatrix}m\\k  \end{pmatrix}p^{n-k}q^{m-k}t^k.
\end{array}
$$
Summing over $n,m$ and using
$\frac{1}{(1-x)^{k+1}}=\sum_{n \geq 0}\begin{pmatrix}n+k \\ k  \end{pmatrix} x^n$,
we obtain 
\begin{eqnarray*}
\sum_{n,m \geq 0} p^n \star q^m&=&\sum_{k,n,m \geq 0}k! \frac{1}{(1-p)^{k+1}} \frac{1}{(1-q)^{k+1}}t^k\\[0.3cm]
&=&\frac{1}{(1-p)(1-q)}E\left(\frac{t}{(1-p)(1-q)}\right).
\end{eqnarray*}
\end{proof}
A similar calculation gives the following slightly more general formula:
\[\frac{1}{1-(\a p+\b q)} \star \frac{1}{1-(\g p+\dt q)}=\frac{1}{\Delta}E(\frac{\a \dt }{\Delta}),\]
where
\[\Delta:=(1-(\a p+\b q))(1-(\g p+\dt q)).\]

These examples show that the product of two meromorphic functions leads to
a series in $t$, that for no fixed values of $p$ and $q$ can be interpreted
as the Taylor expansion of a holomorphic function in $t$ at the origin.

\section{The {\sc Gevrey Heisenberg} algebra and its {\sc Borel} dual}
\label{S::Borel}

Following {\sc Borel}, one may interpret the divergent series that appear in the 
above calculation as the asymptotic expansion of a Laplace integral~\cite{borel1901leccons}.

To do this, in the case of one variable, we first define the {\em Borel transform} of a series $f(t)=\sum_{n \ge 0} a_n t^n$ as the series in a ``Borel-dual'' variable $\xi$ defined by:
\[ g(\xi)=\sum_{n \ge 0} a_n \frac{\xi^n}{n!}.\]
For example, the Euler power series $E(t)=\sum_n n! t^n$ has $g(\xi):=\sum_n\xi^n$ as its Borel transform, which is equal to $\frac{1}{1-\xi}$ if $|\xi|<1$.
If the Borel transform has a positive radius of convergence $R$, for any
$r <R$, one can consider the function
\begin{equation}
F_{r}(t):=\frac{1}{t} \int_0^r g(\xi) e^{-\xi/t}d\xi\;.\label{eq-invbeta}
\end{equation}
The function $F_r$ is holomorphic in the half-plane $\Re(t)>0$, and from the
formula
$$ n!t^n=\frac{1}{t} \int_{0}^{\infty} \xi^n e^{-\xi/t}d\xi, $$
one can show that the function $F_r$ has the series $f(t)$ as asymptotic expansion
on the half-plane: $F_{r}(t) \sim f(t)$.
Note however, that the function $F_{r}$ depends not only on $f$, but also on~
$r$. In particular, to associate a function to the formal power series expansion in this way is, in general,
ambiguous.

\subsection*{The Gevrey-Heisenberg algebra}
Although there is no analytic $\star$-algebra, there is a Gevrey one. In particular, the type of divergence that appeared in the above example
computation of the $\star$-product is typical. We now recall this observation which goes back to {\sc Boutet de Monvel} and {\sc Kr\'ee}~\cite{Boutet_Kree} (see also \cite{Pham_resurgence,Sjostrand_asterisque}).

To do this, we consider the {\em formal Borel transform }  
\[ \beta: \C[[t,q,p]] \lra \C[[\xi,q,p]] \]
defined by setting
\[\beta(\sum_{ijk} a_{ijk}q^ip^jt^k)=\sum_{ijk}a_{ijk}q^ip^j\frac{\xi^k}{k!}.\]
Note that it is a linear bijection that maps $\C[t,q,p]$ onto $\C[\xi,q,p]$,
but, of course, it is not compatible with the product.

As usual, we denote by $\C\{\xi,q,p\}$ the ring of convergent 
power series. A series $f \in \C[[t,q,p]]$ such that 
$\beta(f) \in \C\{\xi,q,p\}$ is called a {\em Gevrey series}. We denote by 
\[\cQ^G:=\{ f \in \C[[t,q,p]]\;\;|\;\;\beta(f) \in \C\{\xi,q,p\} \}\]
the set of all Gevrey series (in $t$, but holomorphic in $q,p$), and we recall the following standard result concerning the $\star$-product.

\begin{proposition}[\cite{Boutet_Kree,Pham_resurgence,Sjostrand_asterisque}]
\label{P::product} The subset $\cQ^G \subset \widehat{\cQ}$  is a subalgebra, i.e., if two functions have a convergent Borel transform, so does their $\star$-product.
\end{proposition}
The algebra $\cQ^G$ was used in \cite{quantique} to prove  a general result saying that the formal Rayleigh-Schr\"odinger series 
for the $n$-th energy level of an anharmonic oscillator are in fact Gevrey series.
\subsection*{The Borel dual algebra}
One can also use the map $\beta$ to transfer the $\star$-product on $\C[[t,q,p]]$ to $\C[[\xi,q,p]]$ and write the Heisenberg algebra in the dual variable $\xi$. So, we introduce the following new product: for any $f,g\in\C[[\xi,q,p]]$,
\begin{equation}
f * g :=\beta(\beta^{-1}(f) \star \beta^{-1}(g)),\label{eq-def*}
\end{equation}
and expand $f=\sum_n \phi_n\xi^n$ and $g=\sum_m \psi_m \xi^m$ in series
with $\phi_n,\psi_m \in \C[[q,p]]$. One can see that the product \eqref{eq-def*} is given by the formula:
\[f * g= \sum_{l\ge 0} \gamma_l \xi^l,\]
where
\begin{equation}
\gamma_l=\sum_{n+m+k=l} \frac{n! m! k!}{(n+m+k)!} \frac{1}{k!}\partial_p^k\phi_n(q,p) \frac{1}{k!}\partial_q^k\psi_m(q,p).\label{eq-formel*}
\end{equation}
This corresponds to the dual version of \eqref{eq-starformel}. Applied to polynomials of $\C[\xi,q,p]$, it gives:
$$q * p=qp,\quad p * q=qp+\xi,\quad \xi^n * \xi^m=\frac{n!m!}{(n+m)!}\xi^{n+m}. $$

Thus, we can directly obtain a dual version of Proposition~\ref{P::product} :
\begin{proposition}
\label{P::Borel}
Consider the non-commutative associative product  on $\C[[\xi,q,p]]$
defined by \eqref{eq-def*}. For any convergent power series $f,g  \in \C\{\xi,q,p\}$, the product $f * g$ is also in $\C\{\xi,q,p\}$.
\end{proposition}
Note that this result can also be derived from the integral formula of the $*$-product given in the next section (Proposition~\ref{P::integral}).\\

We will denote the algebra $\C\{\xi,q,p\}$ with the product $*$ by
$\cQ^B$ and call it the {\em Borel dual algebra}. The formal
Borel transform identifies it with the algebra $\cQ^G$, that is, the 
linear bijection 
$$\beta: \cQ^G \longrightarrow \cQ^B$$
interchanges the $\star$-product on the left hand side with the $*$-product
on the right hand side. 

In going from $\C[[t,q,p]]$ to $\C[[\xi,q,p]]$ with the formal Borel transform, it will sometimes be useful to use the same name for a series in
$\cQ^G$ and its Borel transform in $\cQ^B$ and simply write
$f(\xi,q,p)$ for $\beta(f(t,q,p))$.

\begin{proposition}\label{ex-*} The $*$-product of $1/(1-p)$ and $1/(1-q)$ in $\cQ^B$ is given by the formula
$$ \frac{1}{1-p} * \frac{1}{1-q}=\frac{1}{(1-p)(1-q)-\xi}. $$
\end{proposition}
\begin{proof}
Indeed, we have
$$p^n * q^m=\sum_{k \geq 0}\begin{pmatrix}n\\ k  \end{pmatrix} \begin{pmatrix}m\\k  \end{pmatrix}p^{n-k}q^{m-k}\xi^k. $$
As in the proof of Proposition \ref{P::Euler}, we obtain
$$\sum_{n,m \geq 0} p^n * q^m=\sum_{k \geq 0}\frac{1}{(1-p)^{k+1}} \frac{1}{(1-q)^{k+1}}\xi^k=\frac{1}{(1-p)(1-q)-\xi}.$$
\end{proof}

This proposition has two consequences: it implies Proposition~\ref{P::Euler} and it explains the origin of the divergence for the
$\star$-product. Indeed, choose any $r \in ]0,1[$, Proposition \ref{ex-*} and Formula \eqref{eq-invbeta}  imply that the function 
$$\frac{1}{t}\int_{0}^r \frac{1}{(1-p)(1-q)-\xi}\,  e^{-\xi/t} d\xi$$
has the series given by 
$$\frac{1}{1-p} \star \frac{1}{1-q}, $$ as asymptotic expansion.

As the meromorphic function $\frac{1}{1-\xi}$ is the Borel transform  of the Euler series $E(t)$, the asymptotic expansion at the origin of the right-hand side gives the formula of Proposition~\ref{P::Euler}.

The divergence of the $\star$-product is now explained : it is due to the
appearance of singularities in the dual variable. The ambiguity in the choice of the integration path gives rises to a small exponential
correction for different choices, which cannot be captured by perturbation theory. Let us make this more precise.

The $\star$-product
$$\frac{1}{1-p} \star \frac{1}{1-q} $$ is ambiguous since it defines a divergent series which can be interpreted
as the asymptotic expansion of many holomorphic functions. However from the dual viewpoint, that is for the $*$-product in the $\xi$-variable, there is no longer any ambiguity, and the product is given by:
$$\frac{1}{1-p} * \frac{1}{1-q}=\frac{1}{(1-p)(1-q)-\xi} .$$
As a function of the variable $\xi$, the meromorphic function $\frac{1}{(1-p)(1-q)-\xi}$ is not only holomorphic at zero: it extends to a whole punctured complex
$\xi$-plane with a simple pole at the puncture $\xi=(1-p)(1-q)$.

Let us now slightly deform the half-line going
from $0$ to $+\infty$ to paths
$\gamma_+$ and $\gamma_-$,  in the upper and lower half-plane as in Figure \ref{fig-eulerpaths}.
\begin{figure}[!htb]
  \centering
  \includegraphics[width=10cm]{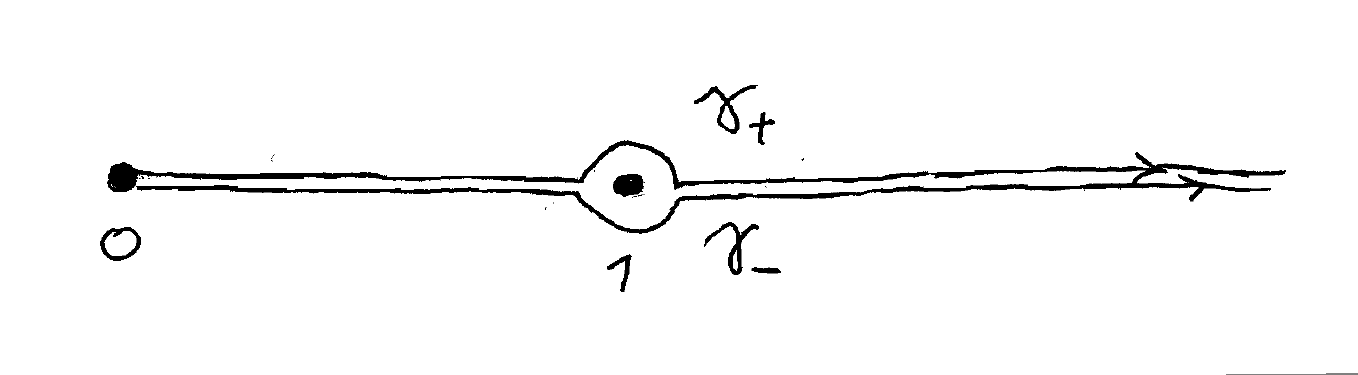}
  \caption[Euler]{\footnotesize{Deformation of $[0,+\infty)$ into $\gamma_+,\gamma_-$.}}
  \label{fig-eulerpaths}
\end{figure}

By following each of these two integration paths, we obtain two preferred  ``Euler functions'' $E_+$ and $E_-$ defined by
\[E_{\pm}(t):=\frac{1}{t}\int_{\gamma_{\pm}}\frac{1}{1-\xi}e^{-\xi/t} d\xi.\]
These are both asymptotic to the Euler series $E(t)$ in the halfplane $\Re(t) >0$ and differ by an
exponentially small function:
\begin{equation}
E_{-}(t)-E_{+}(t)=\frac{1}{t}\int_{\sigma} \frac{1}{1-\xi}e^{-\xi/t}d\xi=\frac{2\pi i}{t}e^{-1/t}.\label{eq-smallexp}
\end{equation}
Here $\sigma$ is a small loop running in the positive direction around the pole at $1$. This small exponential factor explains the
divergence of the original $\star$-product.

Now the important point is that knowing the singularities of $f$ and $g$, we are going to describe the singularities of $f*g$. To this aim, we
first give an integral formula for the $*$-product.

\section{Integral formula for the $*$-product}
Proposition \ref{ex-*} shows that the $*$-product (in the
Borel dual variable $\xi$) can be analytically continued along all
paths that avoid a rather small set.  Our aim is to prove, more generally, that the $*$-product of two multi-valued functions over $\C^{2n+1}$ whose singularity set is algebraic is again a function of this type~(Theorem~\ref{T::product}). This will be done by first proving an explicit integral formula for the $*$-product in this section. To obtain such an integral formula, we start with the case of the $\star$-product (in $t$-variable) and then we look at its integral expression in the Borel plane.
 
\subsection*{The thimble formula}
The following integral expression for the $\star$-product on $\cQ$ is a variant of the Moyal formula \cite{Moyal}:
\begin{proposition}
The $\star$-product of two polynomials $f,g \in \cQ=\C[t,q,p]$ is given by the integral formula
\begin{equation}
f \star g(t,q,p)= \frac{1}{2\pi i t} \int_ \C f(t,q,p+\bar z ) g(t,q+z,p) e^{-| z|^2/t} d\bar z \w dz.\label{eq-starintegr}
\end{equation}
\end{proposition}
\begin{proof}
It suffices to check the formula for $f=p^n$ and $g=q^m$. In the expansion
\[ (p+\bar z)^n(q+z)^m =\sum_{k,l \ge 0} \begin{pmatrix}n\\ k  \end{pmatrix} \begin{pmatrix}m \\ l  \end{pmatrix} p^{n-k} q^{m-l} z^l \bar z^k, \] 
only the terms with $k=l$ will contribute to the integral.
Furthermore, one has
\[ \int_{\C} |z|^{2k} e^{-|z|^2/t}  d\bar z \w dz= 2\pi i k! t^{k+1}\]
and thus it follows from Proposition~\ref{P::Moyal} that we get
indeed the star product $p^n \star q^m$.
\end{proof}
In order to explain the name {\em thimble-formula}, we  will rewrite the above formula \eqref{eq-starintegr} in a slightly more geometrical way.

The  domain integration is the two dimensional chain 
$$D:=\{ (x,y) \in \C^2 :y=\bar x \},$$
so that the formula becomes
\[f \star g(t,q,p)= \frac{1}{2\pi i t} \int_ Df(t,q,p+y) g(t,x+q,p) e^{-xy/t} dx \w dy.\]
This can be  re-written as
\begin{equation}
f \star g(t,q,p)= \frac{1}{2\pi i t} \int_ {D_{q,p}}f(t,q,y) g(t,x,p) e^{-F_{q,p}(x,y)/t} dx\w dy,\label{eq-thimbleform}
\end{equation}
with $F_{q,p}(x,y) := (x-q)(y-p)$ and 
\begin{equation}
D_{q,p}:=\{ (x,y) \in \C^2 :y-p=\overline{ (x-q)} \}.\label{eq-cycle}
\end{equation}
The polynomial $F_{q,p}$ defines a map 
$$\C^2 \lra \C,(x,y) \mapsto F_{q,p}(x,y),$$
that has the point $(q,p)$ as unique non-degenerate critical point with 
critical value $0$. For $\xi \neq 0$, the Riemann surface $X_{\xi,q,p}:=F_{q,p}^{-1}(\xi)$
has the topology of a cylinder and contains a $1$-cycle
$$\gamma_{\xi,q,p}:=D \cap \{F_{q,p}=\xi\}$$ parametrised by 
$\theta \in [0,2\pi]$ via
\[ x(\theta):=q+\sqrt{\xi}e^{i\theta},\;\;y(\theta):=p+\sqrt{\xi}e^{-i\theta}.\]
For $\xi=0$, the cylinder degenerates into a cone~(see Figure \ref{fig-cylinder}).
\begin{figure}[!htb]
  \centering
  \includegraphics[height=6cm]{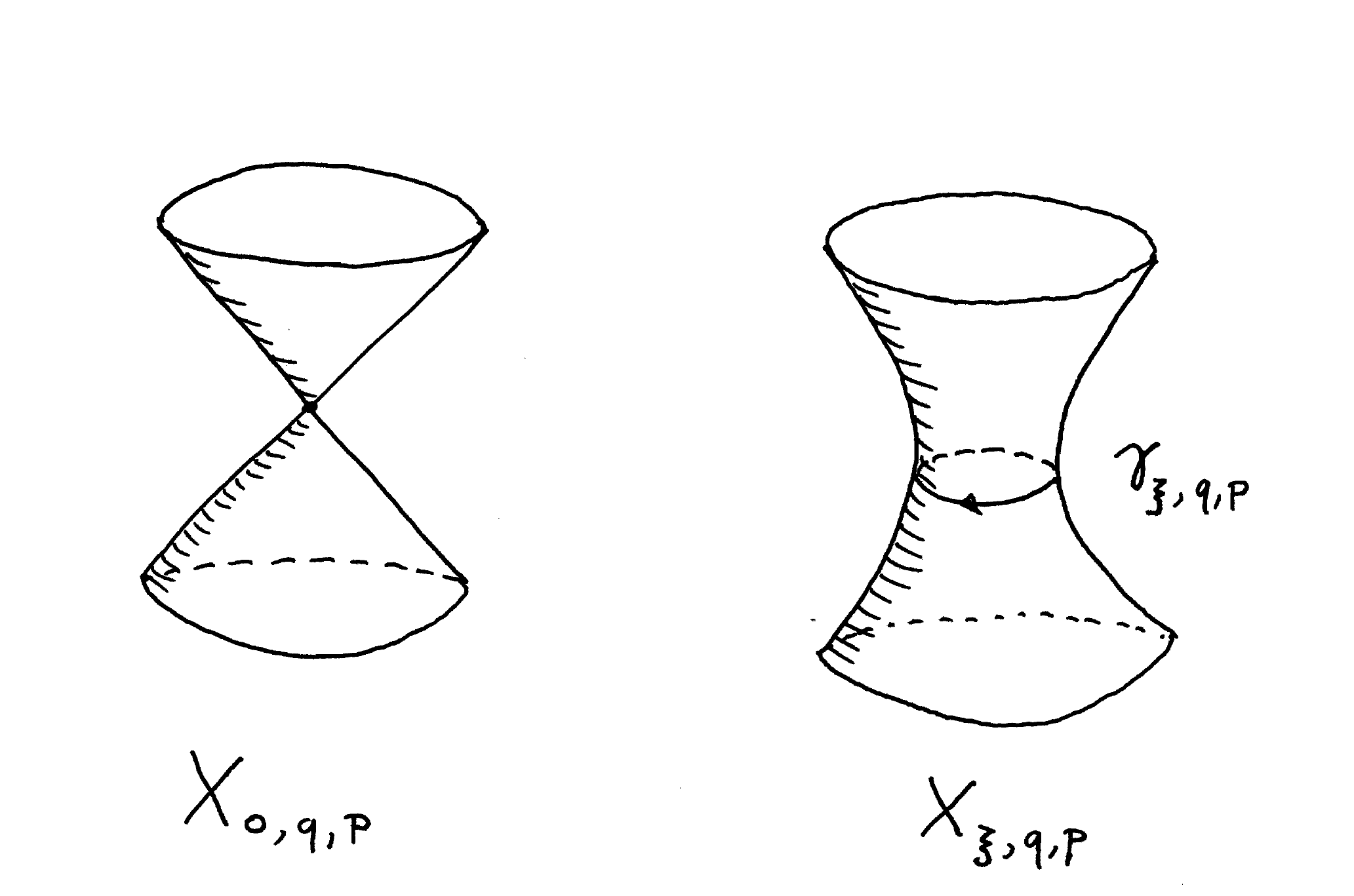}
  \caption[Cylinder]{\footnotesize{Riemann surfaces $X_{0,q,p}$ and $X_{\xi,q,p}$.}}
  \label{fig-cylinder}
\end{figure}

When $\xi$ goes to $0$, the circles $\gamma_{\xi,q,p}$ centred at $(q,p)$ with radius $\sqrt{\xi}$ retract to the critical point $(q,p)$. For this reason, these $\gamma_{\xi,q,p}$ are called  {\em vanishing cycles} 
for the $A_1$-singularity defined by $F_{q,p}$. Note that the cycle $\gamma_{\xi,q,p}$ is a generator of the corresponding homology group
\[ H_1(X_{\xi,q,p})=\Z [\gamma_{\xi,q,p}].\]

In the following, we can restrict to the case $\xi \in \R_{\ge 0}$ useful for the Laplace representation. The cycle $D_{q,p}$ can be seen as a {\em Lefschetz thimble} (see Figure \ref{fig-thimble}), that is, the union of the circles $\gamma_{\xi,q,p}$ centred at $(q,p)$ with radius $\sqrt{\xi}$: 
$$D_{q,p}=\bigcup_{\xi \geq 0}\gamma_{\xi,q,p}. $$

\begin{figure}[!htb]
  \centering
  \includegraphics[width=10cm]{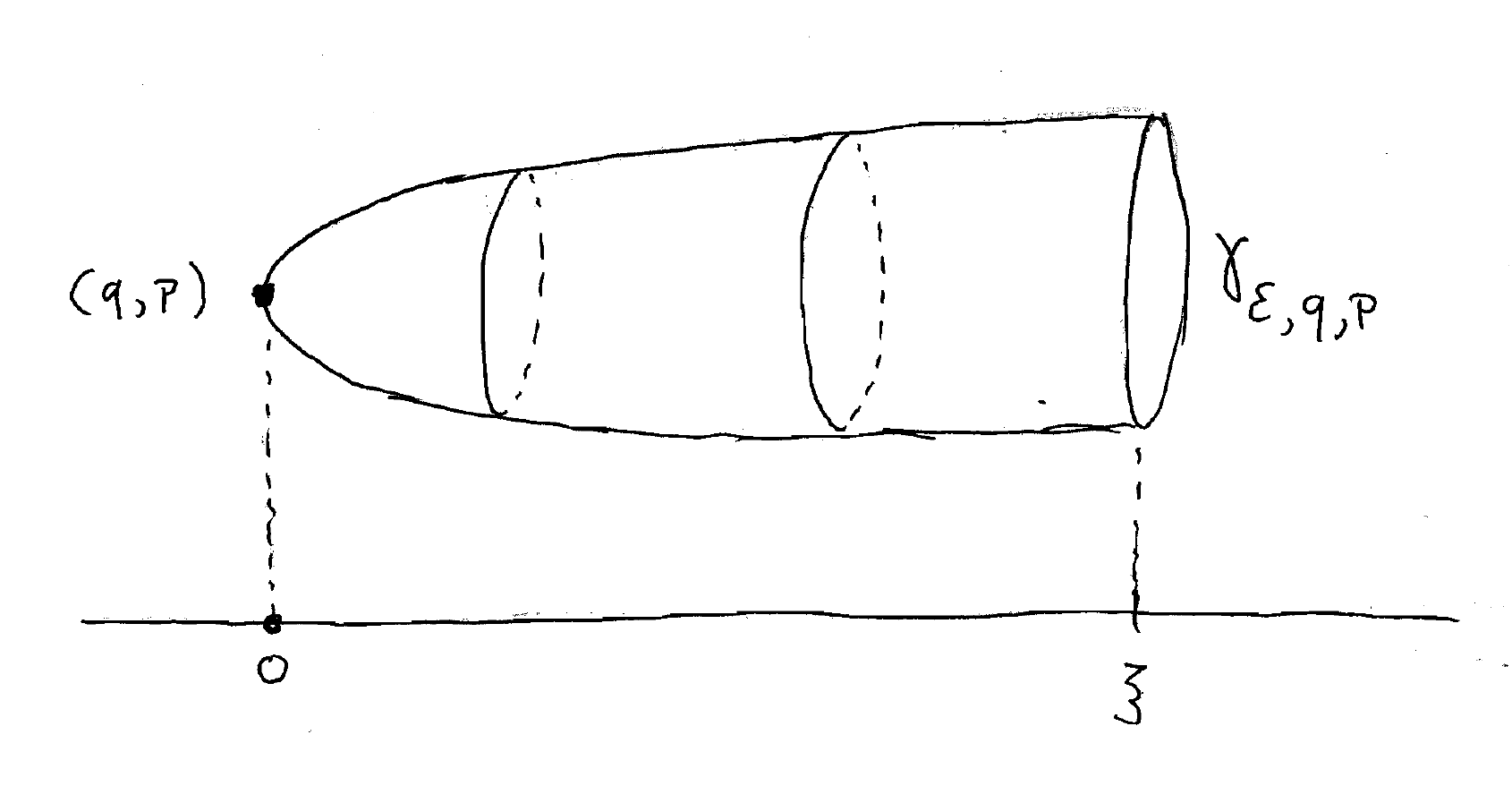}
  \caption[Thimble]{\footnotesize{Lefschetz thimble.}}
  \label{fig-thimble}
\end{figure}

\subsection*{Representation as Laplace integral}
The representation of $D_{q,p}$ \eqref{eq-cycle} as thimble, sliced into vanishing cycles, 
leads to the representation as a Laplace integral.

Consider a polynomial differential two-form $\omega$, so that the integral
$$\int_{D_{q,p}} e^{-F_{q,p}/t}\omega $$
is well-defined for $t\in \R_+$.

Using the general residue-formula 
$$\int_{D_{q,p}} e^{-F_{q,p}/t}\omega =\int_0^{\infty} e^{-\xi/t} d\xi \int_{D_{q,p} \cap \{F_{q,p}=\xi\}} Res \left( \frac{\omega}{F_{q,p}(x,y)-\xi} \right),$$
we can write the integral formula \eqref{eq-thimbleform} for the $\star$-product as:
\[f\star g(t,q,p)=\frac{1}{t} \int_{0}^{\infty} e^{-\xi/t} d\xi \int_{\gamma_{\xi,q,p}} f(\xi,q,y) \bullet g(\xi,x,p) \omega_{\xi,q,p}.\]
 
Let us explain the notations introduced for the integrand. Because we change from $t$ to the Borel dual variable $\xi$, the ordinary
product of functions in $t,q,p$ has to be replaced by the convolution product in the variable $\xi$. It is denoted by $\bullet$ and defined on $f,g\in\C[\xi,q,p]$ by
\[ f(\xi,q,y)\bullet g(\xi,x,p):=\beta(\beta^{-1}(f)(t,q,y).\beta^{-1}(g)(t,x,p)),\]
with explicit expression:
\begin{equation}
p \bullet q=q \bullet p,\quad \xi^n \bullet \xi^m=\xi^n \ast \xi^m=\frac{n!m!}{(n+m)!}\xi^{n+m} .\label{eq-convolpol}
\end{equation}
The holomorphic $1$-form  $\omega_{\xi,q,p}$ on the Riemann surface $X_{\xi,q,p}:=F_{q,p}^{-1}(\xi)$ is defined as the Poincar\'e residue of the $2$-form with first
order pole along the hypersurface $X_{\xi,q,p}$:
\[ \omega_{\xi,q,p}:=\frac{1}{2\pi i} Res\left(\frac{dx\wedge dy}{F_{q,p}(x,y)-\xi}\right).\]
A representative for  $\omega_{\xi,q,p}$ can be computed explicitly as
$$Res\left(\frac{dx\w dy}{F_{q,p}(x,y)-\xi}\right)=\frac{dx}{x-q} $$
with $\int_{\gamma_{\xi,q,p}}\omega_{\xi,q,p}=1$. This explains the expression of the star product as a Laplace integral. Our aim is to extend this integral formula to a larger class of function than polynomials but,
before that, we indicate how to generalise the formula to higher dimensions.

\subsection*{Extension to $n$-degrees of freedom}
The above discussion can be generalised to $n$ degrees of freedom. 
For $$(q,p)=(q_1,q_2,\ldots,q_n,p_1,p_2,\ldots,p_n),$$ we consider the
polynomial
$$F_{q,p}(x,y)=\sum_{j=1}^n(x_j-q_j)(y_j-p_j). $$
It defines a map $\C^{2n} \lra \C$ which has $(q,p)$ as unique
non-degenerate critical point, i.e. an $A_1$-singularity in $2n$-variables.
 
The complex $(2n-1)$-dimensional hypersurface $X_{\xi,q,p}=F^{-1}_{q,p}(\xi)$ contains a real $(2n-1)$-dimensional vanishing sphere
\[\gamma_{\xi,q,p}=(q,p)+\{(z,\bar z):| z_1|^2+\dots+|z_n|^2=\xi \} .\]
By orienting this sphere, we get a generator of the middle dimensional homology group:
\[ H_{2n-1}(X_{\xi,q,p})=\Z [\gamma_{\xi,q,p}],\]
 for $\xi \neq 0$. The hypersurface $ X_{\xi,q,p}$ carries a holomorphic $(2n-1)$-form
$$\omega_{\xi,q,p}:=\frac{1}{(2\pi i)^n}\, Res \left( \frac{dx_1\w dy_1\w \dots \w dx_n \w dy_n}{F_{q,p}-\xi} \right).$$
One can also easily compute a representative for the residue form
$$ Res \left( \frac{dx_1\w dy_1\w \dots \w dx_n \w dy_n}{F_{q,p}-\xi} \right)= \frac{dx_1\w dy_1\w \dots \w  dx_{n-1}  \w dy_{n-1} \w dx_n}{x_n-q_n}.$$
The sphere $\gamma_{\xi,q,p}$ is oriented in such a way that $\int_{\gamma_{\xi,q,p}} \omega_{\xi,q,p} = 1$.
As in the case of $n=1$, by denoting also $\cQ=\C[t,q,p]$ (with $q,p\in\C^n$) endowed with the $\star$-product, we find:
\begin{proposition}
For $f,g \in \cQ$ one has
\[f \star g (t,q,p)=\frac{1}{t} \int_0^{\infty} e^{-\xi/t} \left( \int_{\gamma_{\xi,q,p}}f(\xi,q,p) \bullet g(\xi,q,p) \omega_{\xi,q,p} \right) d\xi.\]
\end{proposition}

\subsection*{Vanishing cycle formula}
The above representation of the $\star$-product as Laplace integral
suggests that it is possible to express the $*$-product in the Borel dual
$\xi$ variable directly as an integral over the vanishing cycle. This is one of the key results of this paper and it turns out that the formula makes sense for arbitrary elements of $\cQ^B$. 
\begin{proposition}
\label{P::integral} For $f,g \in \cQ^B=\C\{ \xi,q_1,\dots,q_n,p_1,\dots,p_n\}$ 
the $*$-product is expressed  into an integral of the $\bullet$-product over a vanishing cycle
\begin{equation}
f * g (\xi,q,p)=\int_{\gamma_{\xi,q,p}} f(\xi,q,y) \bullet g(\xi,x,p) \omega_{\xi,q,p},\label{eq-integr*}
\end{equation}
where $ \gamma_{\xi,q,p}$  is the $(2n-1)$-dimensional sphere 
$$ \gamma_{\xi,q,p}=(q,p)+\{(z,\bar z):| z_1|^2+\dots+|z_n|^2=\xi \} ,$$ 
 $(\xi,q,p)$ belongs to a sufficiently small neighbourhood of  the origin and $\xi \in \R_{>0}$.
\end{proposition}
Before giving the proof, we remark that the $\bullet$-product \eqref{eq-convolpol} can be extended on 
two elements from $\C\{\xi,q,p\}$ and it obviously again belongs to 
$\C\{\xi,q,p\}$. Thus it follows from the formula \eqref{eq-integr*} that the non-commutative algebra
$\cQ^B$ is closed under $*$~(Proposition~\ref{P::Borel}).

\begin{proof}
 When we expand both sides of the to-be-proven equality
\[ f \ast g (\xi,q,p)=\int_{\gamma_{\xi,q,p}} f(\xi,q,y) \bullet g(\xi,x,p) \omega_{\xi,q,p}\]
in powers of $\xi$, since $\xi^n\ast\xi^m=\xi^n\bullet\xi^m$ (see \eqref{eq-convolpol}), we readily reduce to the case when $f$ and $g$ do not 
depend on $\xi$. 

Next, we fix $(q,p)=(q_1,q_2,\ldots,q_n,p_1,p_2,\ldots,p_n) \in \C^{2n}$ and consider Taylor expansions at the origin  of the functions $y \mapsto f(q,y)$ and $x \mapsto g(x,p)$. We get:
\[g(x,p)=\sum_\a a_\a(q,p)(x-q)^\a,\;\;\  f(q,y)=\sum_\b b_\b(q,p) (y-p)^\b\]
where $\a=(\a_1,\dots,\a_n)$ and $\b=(\b_1,\dots,\b_n)$ are multi-indices and
\[ a_\a(q,p):=\frac{1}{(\sum_{j=1}^n\a_j)!}\partial^\a_pg(q,p),\;\;\;b_\b(q,p):=\frac{1}{(\sum_{j=1}^n\b_j)!}\partial_q^\b f(q,p).\]
As the cycle $\gamma_{\xi,q,p}$ is compact, we can interchange the integral and summation:
$$\int_{\gamma_{\xi,q,p}} f(q,y) \bullet g(x,p) \omega_{\xi,q,p} =\sum_{\a,\b}a_\a b_\b \int_{\gamma_{\xi,q,p}} (x-q)^\a (y-p)^\b \omega_{\xi,q,p}.$$
Therefore, according to the formula \eqref{eq-formel*} of the $*$-product in the Borel dual algebra, the above proposition reduces to the following lemma:
\begin{lemma}\label{L::beta} For any $\a,\b \in \Z^n_{\geq 0}$, we have
\[ \int_{\gamma_{\xi,q,p}} (x-q)^\a (y-p)^\b \omega_{\xi,q,p}=\frac{\prod \a_j!}{(\sum \a_j)!}\dt_{\a,\b}\xi^{|\alpha|}\]
with $|\a|:=\sum_{j=1}^n\alpha_j$.
\end{lemma}
\begin{proof}
As the left and the right-hand side are invariant under translation, it is sufficient to prove the lemma for $q=p=0$. By homogeneity, we may also assume that $\xi=1$. 

We now compute explicitly the integral for $q=p=0$, $\xi=1$. To do this, we parametrise the sphere $\g_{\xi,q,p}$ by
$$x_j=\sqrt{s_j}e^{i\p_j},\ y_j=\sqrt{s_j}e^{-i\p_j}, $$
where $(s_1,s_2,\ldots,s_n)$ belongs to the simplex $\Delta \subset \R^n$ defined by the conditions $s_j \ge 0$, $\sum_j s_j = 1$.
We get
$$\frac{dx_1\w dy_1\w \dots \w dx_n \w dy_n}{x_1y_1+\dots+x_ny_n-\xi}=\frac{ds_1\w d\p_1\w \dots\ ds_n \w d\p_n}{s_1+\dots+s_n-\xi},$$ 
so
$$Res\left(\frac{dx_1\w dy_1\w \dots \w dx_n \w dy_n}{x_1y_1+\dots+x_ny_n-\xi}\right)= d\p_1\w ds_2 \w d\p_2 \w \dots\w ds_n \w d\p_n, $$
and
$$ \int_{\gamma_{1,0} }x^\a y^\b\omega_{\xi,0}=\dt_{\a,\b}\int_{\Delta} s^\a ds_2 \w \dots \w ds_n .$$
This integral over the simplex is well-known; it is a case of the Dirichlet multi-dimensional generalisation of the beta-integral of Euler: 
$$ \int_{\Delta} s^\a ds_2 \w \dots \w ds_n=\frac{\prod \a_j!}{(\sum \a_j)!} .$$
\end{proof}
 \end{proof}

\section{Analytic continuation}
\label{sec-anal}
From the integral formula of Proposition \ref{P::integral}, we see that the analytic continuation of the
$\ast$-product naturally falls into two sub-problems:
\begin{enumerate}[{\rm A)}]
\item study the continuation properties of integrals of the form
\[\int_{\gamma_{\xi,q,p}} f \omega_{\xi,q,p},\] 
\item study the continuation properties of the $\bullet$-product.
\end{enumerate}
The computation of the singularities for problems A) and B) determines the singularities of $f*g$ as a particular case.

\subsection*{Riemann domain and analytic continuation.}
 The analytic continuation of a holomorphic function germ $f \in \C\{x_1,x_2,\ldots.x_n\}$ along a path $\gamma$ starting at $0$ may be blocked by a singularity. Sometimes one may deform $\gamma$ slightly to circumvent it and resume the 
continuation. In other cases an essential boundary appears and such a continuation becomes impossible.
\begin{figure}[!htb]
  \centering
  \includegraphics[scale=0.4]{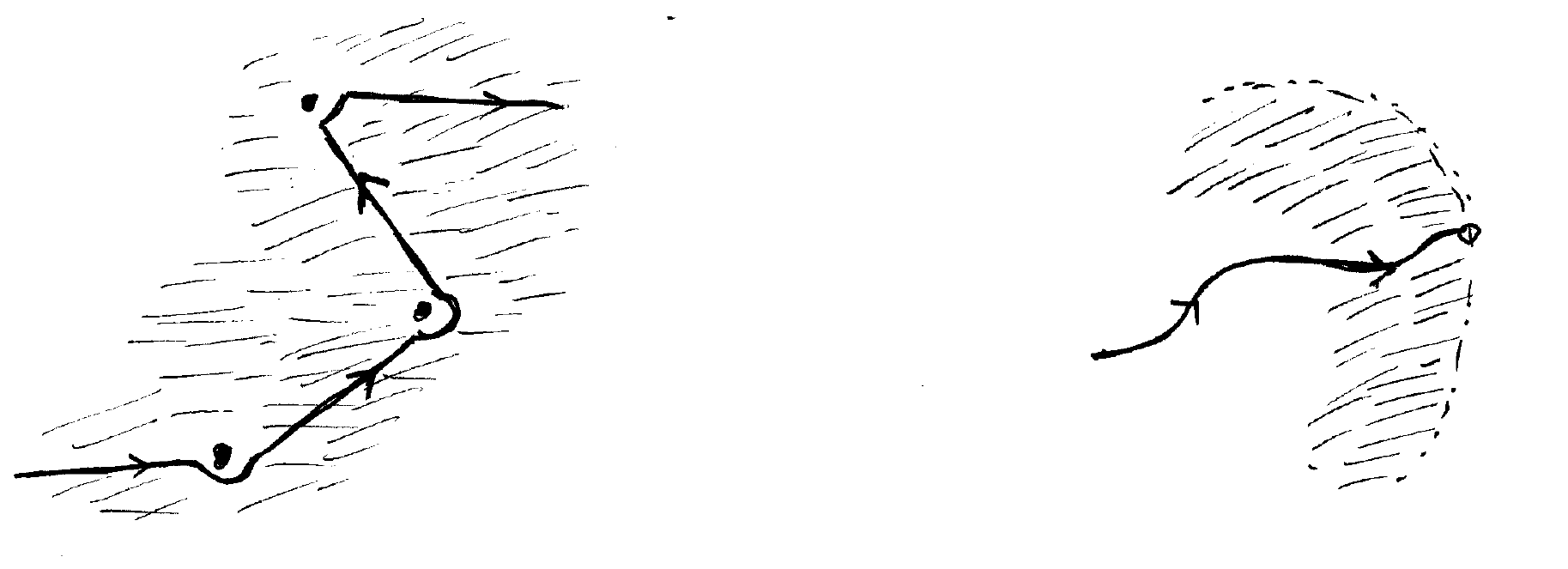}
  \caption[Analytic]{\footnotesize{Analytic continuation of a holomorphic germ.}}
  \label{fig-analytic}
\end{figure}

For example,  in one variable, the first alternative occurs for the power series expansions of $(1-x)^{\alpha},\ \a \in \C$ or $\log(1-x)$, whereas the $\theta$-series $\sum_{n=0}^{+\infty} x^{n^2}$ provides an
example of the second type of behaviour: it cannot be extended analytically 
outside the unit disk. 

The notion of Riemann surface attached to a germ extends naturally to arbitrary dimensions~:
continuations along different paths with the same endpoint may lead to different
values, but the set of all continuations of a germ $f \in \C\{x_1,x_2,\ldots,x_n\}$ can be made into a (connected) $n$-dimensional complex manifold $R_f$, called the {\em Riemann domain}, on which it has a single-valued extension~(see for instance~\cite[Chapter III]{Chabat_Introduction_II}).

 This Riemann domain  comes with a natural projection map
\[ \pi : R_f \lra \C^n,\]
that is locally biholomorphic and has discrete fibres.  The germ $f$ itself represents a canonical origin $O \in R_f$ lying over the origin in $\C^n$. As a general rule, the map $\pi$ will, however,  not be a (regular) covering in the topological sense.

A path $\gamma:[0,1] \lra \C$ has at most one lift to a path in $R_f$ starting at $O$. A path that is {\em not} liftable to a path starting at $O$, but whose restriction to $[0,1)$ {\em is} liftable, is called a {\em blocked path} and its endpoint $\gamma(1) \in \C^n$ is called a {\em singular point of $f$}.
We denote by $\Sigma_f \subset \C^n$ the set of all singular points of $f$. 
Clearly, $f$ can be continued along any path that avoids the set $\Sigma_f$.

In general, even for $n=1$, the structure of the map $\pi:R_f \lra \C$ and the set $\Sigma_f$ can be extremely complicated.
In the simplest cases the singular set $\S_f$ is finite. This happens for instance if $f$ is algebraic, or more 
generally, if $f$ is {\em holonomic}, i.e. satisfies a homogeneous linear differential equation with polynomial 
coefficients.  Slightly more complicated are the cases in which $\S_f$ is countable and discrete. There are however many
important germs not belonging to this class. For example, the inverse function of the indefinite Abelian integral 
\[ S(x)=\int_0^x p dq,\;\;p^2-F(q)=0,\] 
where $F$ is a general polynomial of degree $\ge 5$, provides an example of a germ 
for which $\Sigma_f$ is a countable {\em dense} subset~\cite{Pham_livre_resurgence}. 
Far worse behaviour can occur: in $1918$ {\sc Gross} gave an example of an
entire function $g:\C \lra \C$ which has every value as asymptotic value~\cite{gross1918ganze} .
If $f$ denotes the germ at $0$ of the inverse of $g$, one can identify the map
$\pi: R_f \lra \C$ with $g:\C \lra \C$, and $\Sigma_f=\C$.

\subsection*{Algebro-resurgence.}
One key idea of resurgence theory, developed first by {\sc \'Ecalle} \cite{Ecalle_fonctions} and then by {\sc Pham} \cite{Pham_livre_resurgence}, is to single out classes of germs $f \in \C\{x_1,x_2,\ldots,x_n\}$ closed under interesting operations 
like convolution product and for which the singular set $\S_f$ is not too big.

The weakest condition is maybe to ask that $\C^n\setminus \Sigma_f$ is path-connected and dense. In such a situation $f$ 
has the {\em Iversen property}: for each path $\gamma$ starting at $0$ and each $\e >0$, there is an $\e$-near 
path $\tilde{\gamma}$ along which one can continue $f$~\cite{eremenko2004geometric}.

A stronger natural condition is to ask that $\Sigma_f$ is a countable union of  (algebraic or analytic) hypersurfaces.
This gives a variant of resurgence, that we  call {\em algebro-resurgence}~:
\begin{definition}\label{D::ar} We say that $f \in \C \{ x_1,\dots,x_n \}$ is 
  algebro-resurgent if  $\Sigma_f$ is an algebraic subvariety of $\C^n$.
\end{definition}
Algebro-resurgent power series of one variable have a finite singular set $\S_f$~: meromorphic functions, fractional powers, logarithms, algebraic functions, solutions of linear differential equations with regular poles are algebro-resurgent. But the gamma function and most indefinite abelian
integral  are not algebro-resurgent.

We may now state our main result~:
 \begin{theorem}~\label{T::product} The non-commutative associative product $f * g$ of two algebro-resurgent power series
 $f,g \in \C\{ \xi,q_1,\dots,q_n,p_1,\dots,p_n\}$ is also alge\-bro-re\-surgent. 
\end{theorem}
As we shall now see, the theorem is a consequence of the integral formula for the $*$-product and standard stratification theory. The proof is constructive~: knowing the singularity sets of $f$ and $g$, it gives an explicit description of the singularity set for $f*g$.
\subsection*{Stability under integration.}
To fix the ideas, let us first come back to the integral formula \eqref{eq-integr*} with one degree of freedom. So let us assume for the moment that problem B) is solved and that
$$h(\xi,q,p,x,y)=f\bullet g (\xi,q,p,x,y)$$
is an algebro-resurgent power series.  We denote by  $\S_h$ its singular locus, which is thus supposed to be an algebraic $4$-fold in $\C^5$.

 The Riemann surface $X_{\xi,q,p}$  intersects $\S_h$ in finitely many points. 
 If $(\xi,q,p)$ is sufficiently close to the origin then these points are far away from the vanishing cycle
$\gamma_{\xi,q,p}\,$, so that the $*$-product is well-defined by the integral formula \eqref{eq-integr*}.

  As one moves $(\xi,q,p)$ further from the origin,
these intersection points start ``moving around'' on the Riemann surface, and one has to continue the vanishing cycle avoiding the moving points. In such a case, the vanishing cycle $\gamma_{\xi,q,p}$ separates the Riemann surface in two components and hence the singular points in two groups. 

Now, when two points on different sides of $\gamma_{\xi,q,p}$ come together, the cycle gets pinched, the integral develops a  singularity and the cycle cannot avoid the singular points any longer (see LHS of Figure \ref{fig-pinch}). This corresponds to a singularity of the integral and hence of the $*$-product.
\begin{figure}[!htb]
  \centering
  \includegraphics[height=6cm]{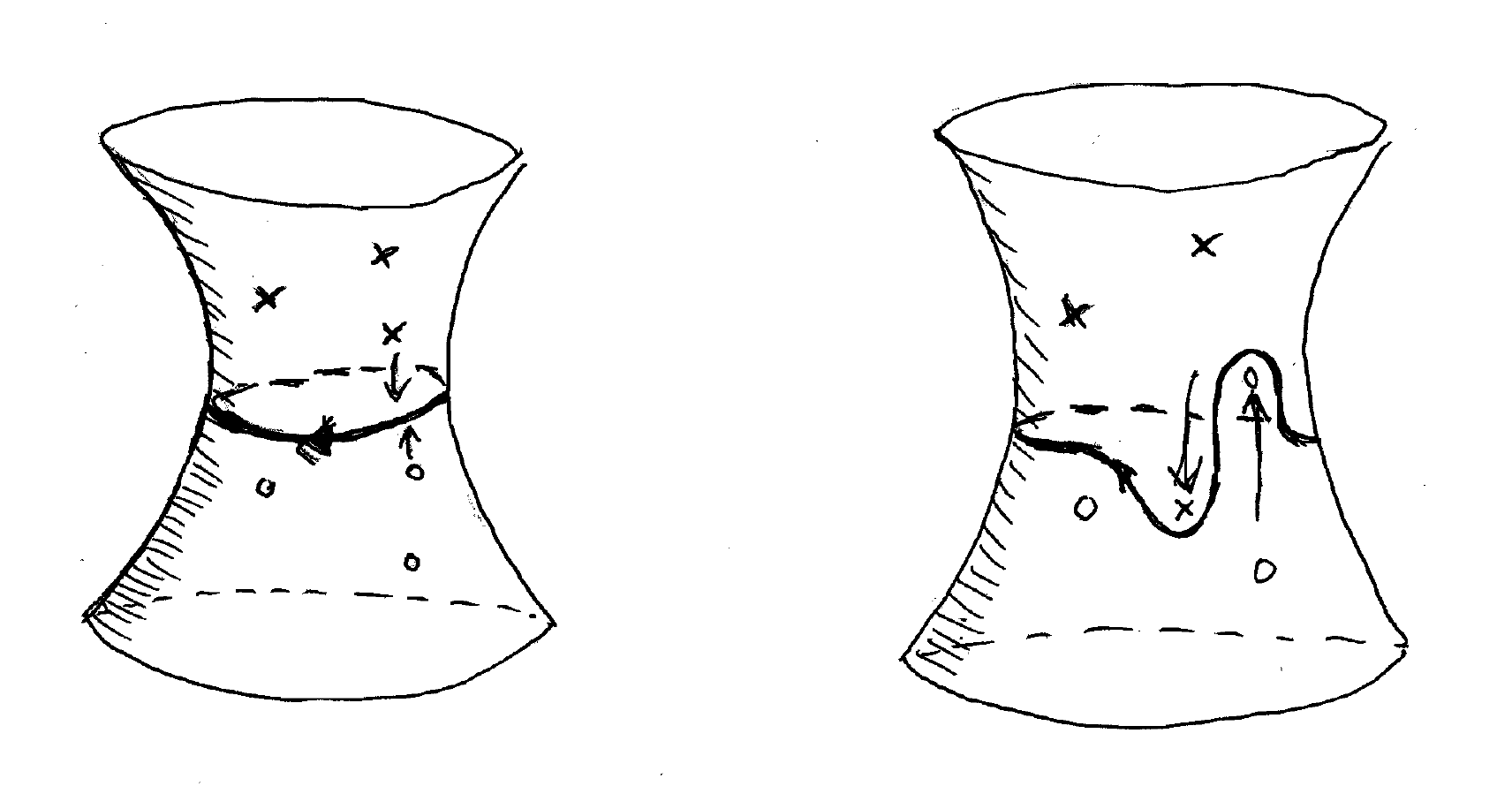}
  \caption[Pinch]{\footnotesize{Singular points on $X_{\xi,q,p}$.}}
  \label{fig-pinch}
\end{figure}
Another thing that may happen is that some $(\xi,q,p)$, one of the singular points `runs to infinity' and pushes the 
cycle with it (see RHS of Figure \ref{fig-pinch}). However, as long as one avoids such a collision and run-away catastrophe, the cycle can be deformed as 
to stay away from the singularities and the function can be in this way analytically continued. 

This situation is of course general and holds for the {\em integral of any closed algebro-resurgent differential form
over a cycle.} By algebro-resurgent differential $p$-form on $X=\C^n$, we mean a germ of a $p$-form $\omega$
for which the coefficients $A_I=A_I(x_1,\ldots,x_n)$ in a local coordinate representation
\[\omega=\sum A_I dx_I,\;\;dx_I=dx_{i_1}\wedge dx_{i_2} \wedge \ldots \wedge dx_{i_p},\]
are all algebro-resurgent germs. The singular locus $\Sigma_{\omega}$ is defined to be the union of singular loci of the coefficients.

For a polynomial map $F:X \lra \C^l$, one means by {\em horizontal family of $p$-cycles} over $V \subset \C^l$, a section over $V$ of the direct image sheaf  $R^pF_*\Z$. This is the sheaf associated to the presheaf 
$$V \mapsto H^p(F^{-1}(V),\Z),$$ and a section over $V$ can be thought of as a family $\Gamma_{\l}, \lambda \in V$, of cycles in the fibres $X_{\l}:=F^{-1}(\l)$ of the map $F$.

\begin{proposition} Let $\omega$ be a closed algebro-resurgent $p$-differential form on $X=\C^N$ with singularity set $\S_{\omega}$ and let
\[F:X \to \C^l \]
be a polynomial map. Let $\Gamma_{\l}, \l \in V$, a horizontal family of $p$-cycles in $X \setminus \S_{\omega}$ over an open neighbourhood $V \subset \C^l$ of 
the origin. Then the germ at $0$ of the function on $V$ defined by the integral
$$g(\l)=\int_{\Gamma_{\l}}  \omega$$
is algebro-resurgent.
\end{proposition}
\begin{proof}
It is a fundamental fact from affine algebraic  geometry that there exists 
a Zariski-open subset $U \subset \C^l$ such that the restriction 
\[F':(X \setminus \Sigma_\omega) \cap F^{-1}(U) \lra U \]
of $F$ over the set $U$ is a topologically trivial fibration~(see for instance~\cite{Verdier_w}). 

Consider a path $\gamma:[0,1] \lra \C^l$ with $\gamma(0)=0$ and whose
restriction to $(0,1]$ is mapped into $U$. 
By the local topological triviality over $U$, we can continue the horizontal 
family of cycles $\Gamma_{\l}, \l \in U \cap V$, along the path $\gamma$. By 
construction, the continuation of the cycle $\Gamma_{\l}$ stays inside 
$X \setminus \Sigma_\omega$ and thus the differential form $\omega$ can be 
continued along the trace of the cycle. This shows that the germ $g(\l)$ can be analytically 
continued along all paths starting at $0$ and (whose restriction to $(0,1]$)
avoid the algebraic set $\C^l\setminus U$.
\end{proof}
Note that the above proof is constructive: the singularities of the integral $g(\l)=\int_{\Gamma_{\l}}  \omega$ are explicitly described once we chose the corresponding fibration. 

We apply the proposition to the polynomial map
\[F:\C^{4n} \lra \C^{2n+1}, \;(q,p,x,y) \mapsto (\sum_i^n(x_i-q_i)(y_i-p_i),q,p),\]
which is the composition of
\[ \C^{4n} \lra \C^{4n+1}, (q,p,x,y) \mapsto (\sum_{i=1}^n(x_i-q_i)(y_i-p_i), q,p,x,y)\]
with the canonical linear projection
\[ \C^{4n+1} \lra \C^{2n+1},\;(\xi,q,p,x,y) \mapsto (\xi,q,p)\;.\]
We take also the family of vanishing  cycles $\gamma_{\xi,q,p} \in  H_n(X_{\xi,q,p},\Z)$
as the horizontal family. So, to conclude the proof of Theorem~\ref{T::product}, it remains to prove that the product $f\bullet g$ of algebro-resurgent functions is also algebro-resurgent. Let us first analyse the additive convolution.

\subsection*{Stability under additive convolution.}
The behaviour of the singular set under convolution is a classical subject of 
analysis, which goes back to the papers of {\sc Hadamard} and {\sc Hurwitz}~\cite{Hadamard_produit,Hurwitz_produits}. The {\em Hadamard product} of
two formal power series
 $$f =\sum_n a_n\xi^n,\quad g=\sum_n b_n \xi^n \in \C[[\xi]]$$
is defined as $\sum_n  a_nb_n \xi^n$. If $f$ and $g$ are convergent power series, it can be represented by the
 integral formula
\[\frac{1}{2\pi i} \oint f\left(t \right) g\left(\frac{\xi}{t}\right)\frac{dt}{t},\]
called multiplicative convolution. Similarly, the {\em Hurwitz product }
\[\sum_k  c_k \xi^k,\ c_k:=\sum_{n+m+1=k} \frac{n! m!}{(n+m+1)!}a_nb_m,\]
can be expressed as  {\em  the additive convolution} 
of $f$ and $g$:
\begin{equation}
f \oplus g:= \int_0^{\xi} f(t)g(\xi-t)dt.\label{eq-addconvol}
\end{equation}
(We use neither the notation $*$ nor $\star$ for the convolution product to avoid confusions with previous products.)

These integral formulas can be used to show that the singularities of the convolution are obtained by multiplication resp. addition of the singularities of $f$ and $g$. This result will be useful for the case of the $\bullet$-product.
\begin{proposition} 
\label{P::closed}
Let $f,g \in \C\{ x \}$ be two algebro-resurgent functions. The  additive convolution $f \oplus g$ is also algebro-resurgent and its
singularity set is a subset of 
$$\left( \S_f + \S_g \right) \cup \S_f \cup \S_g.$$
 \end{proposition}
\begin{proof}
Each of the functions $f,g$ possesses a Riemann surface $R_f$, $R_g$ together with a projection
$$R_f \stackrel{\pi_f}{\longrightarrow} \C,\ R_g \stackrel{\pi_g}{\longrightarrow} \C, $$
which combine to a map $\pi: R_f \times R_g \lra \C \times \C$. The {\em sum map} $\C \times \C \to \C,\ (x,y) \mapsto x+y $, pulls-back to a map on the product $R_f\times R_g$:
 $$R_f \times R_g \to \C,\ (x,y) \mapsto \pi_f(x)+\pi_g(y). $$

Now consider a path $\g :[0,1] \to \C$ whose image avoids both $\S_f$ and $\S_g$. It lifts to both Riemann surfaces, so we get paths $\g_f$ in $R_f$ 
and $\g_g$ in $R_g$.

By the Poincar\'e-Leray residue formula, for $\xi=\g(t)$, the convolution product is given by the formula
$$f \oplus g(\xi)=\int_{ \dt_t} f(x)g(y) Res\left(\frac{dx \w dy}{x+y-\xi}\right), $$
where $\dt_t$ is a path joining $(\g_f(t),0)$ to $(0,\g_g(t))$ in the fibre 
$$\{ \pi_f(x)+\pi_g(y) =\xi \} \subset R_f \times R_g,$$ depending continuously on $\xi$.
As the integral of a holomorphic differential form along a continuous family of chains is holomorphic, analytic continuation reduces to a topological issue: to find paths $\dt_s$ on $R_f \times R_g$, depending continuously on $s$, which connect  $(\g_f(s),0)$ to $(0,\g_g(s))$ and such that the path $ \dt_s $ projects to the point $\g(s) \in \C$. See Figure \ref{fig-produit} for a real picture in case the Riemann surfaces of $f$ and $g$ are respectively $\C \setminus \{ \a \}$ and $\C \setminus \{ \b \}$ with $\a,\b \in \R$.
\begin{figure}[!htb]
  \centering
  \includegraphics[scale=0.4]{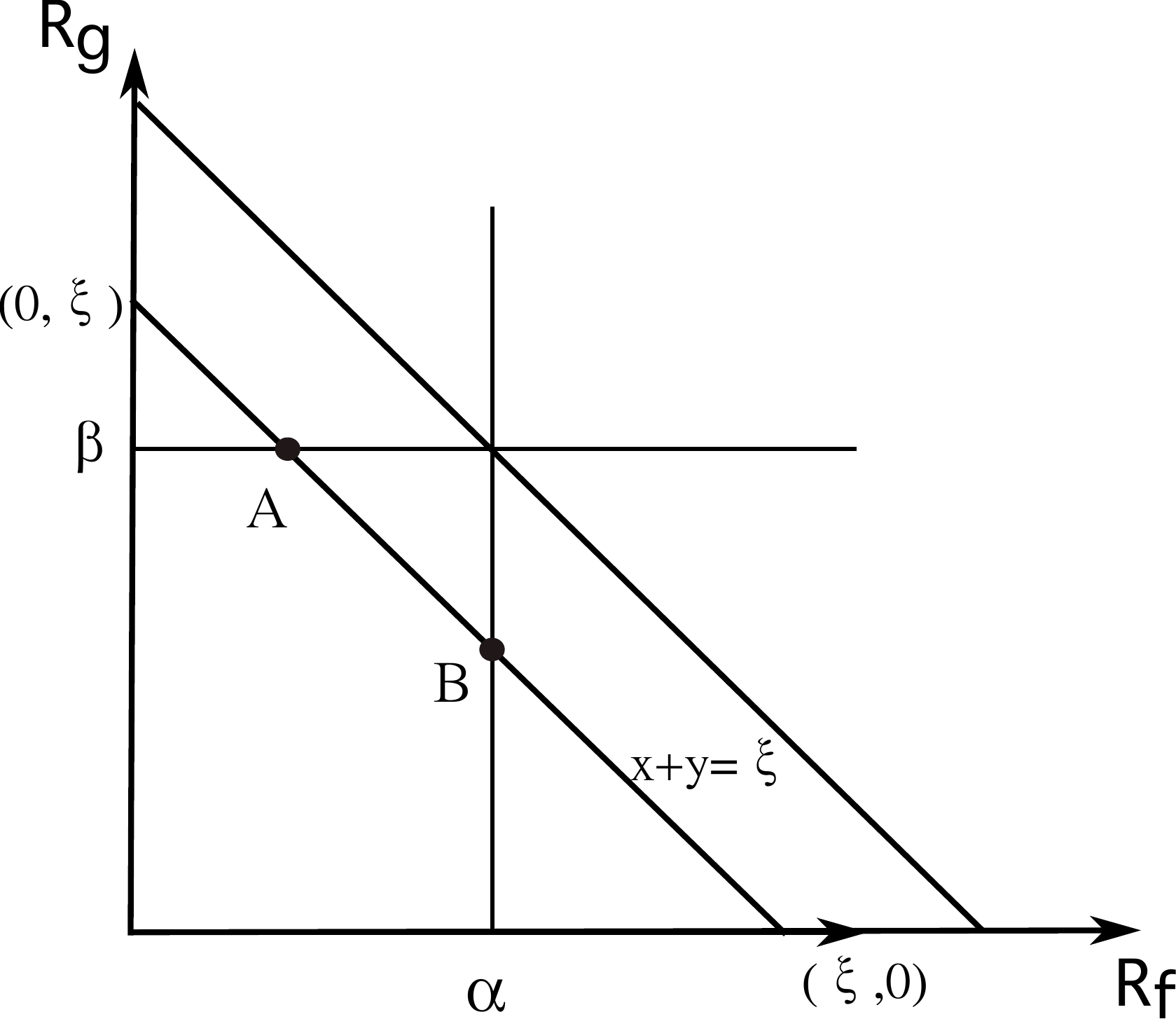}
  \caption[Produit]{\footnotesize{Real picture of the path $\delta_s$: $x+y=\xi$.}}
  \label{fig-produit}
\end{figure}

There is an obvious obstruction extending a lift $\dt_s$ : if $\xi=\g(s)$ is of the form $\a+\b$ with $\a \in \S_f$ and $\b \in \S_g$, the path $\dt_s$ might get pinched as in Figure \ref{fig-pintch}.
\begin{figure}[!htb]
  \centering
  \includegraphics[scale=0.3]{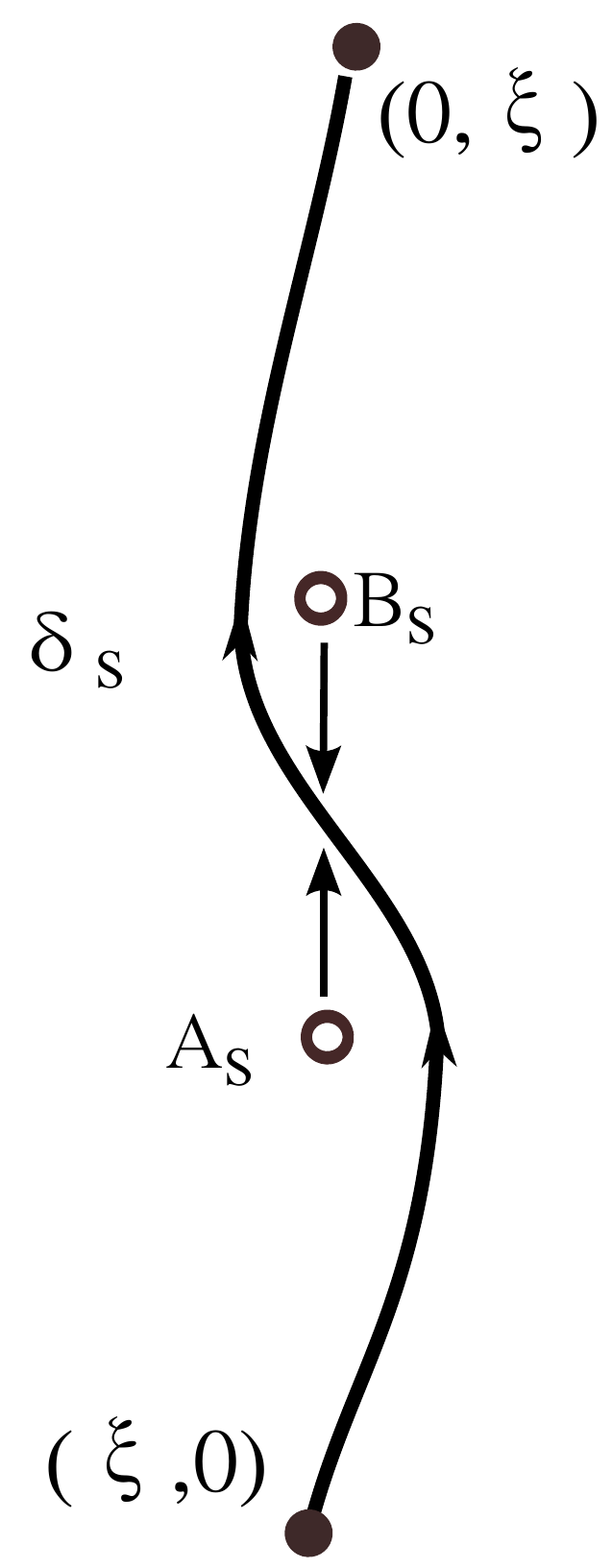}
  \caption[Pintch]{\footnotesize{Map $\delta_s$ getting pinched by singular points $A_s$ and $B_s$.}}
  \label{fig-pintch}
\end{figure}
If we turn around the point $\a+\b$, then we may continue the path as in Figure \ref{fig-avoid}.
\begin{figure}[!htb]
  \centering
  \includegraphics[scale=0.5]{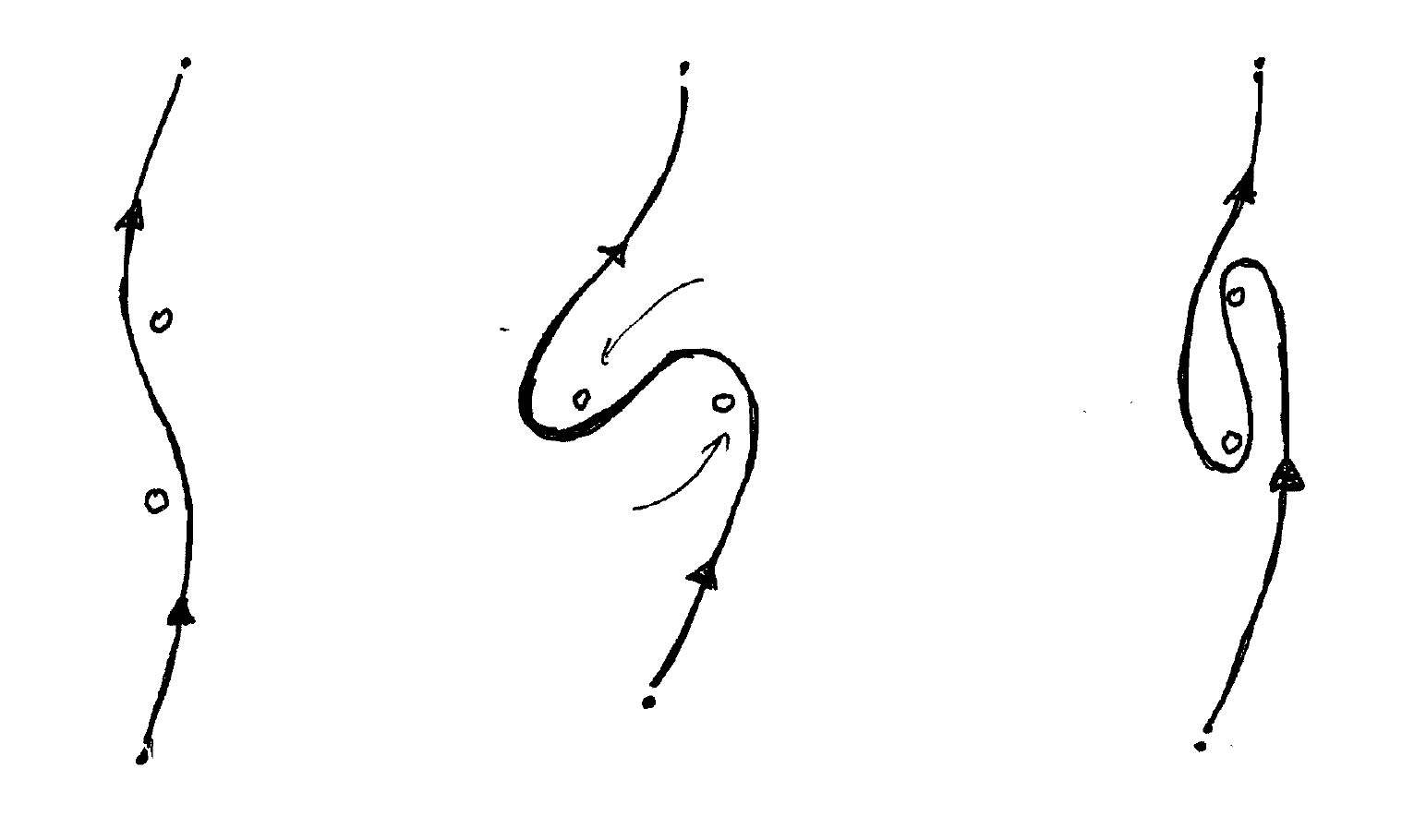}
  \caption[Avoid]{\footnotesize{Analytic continuation when $\xi$ avoids the set $\Sigma_f+\Sigma_g$.}}
  \label{fig-avoid}
\end{figure}
 
This explains that analytic continuation can only be ensured  if we also avoid the set
$$ \S_f + \S_g:=\{ \a+\b,\ \a \in \S_f,\ \b \in \S_g\}.$$
The maps $\pi_f$ and $\pi_g$ are local diffeomorphisms thus, according to the above discussion, the following lemma finishes the proof.

\begin{lemma}
\label{P::lift}
If $A$ and $B$ are finite subsets of $\C$,  the sum map induces a locally trivial fibration
$$ \C \setminus A \times \C \setminus B \to \C \setminus \left(A \cup B \cup (A+B)\right),\ (x,y) \mapsto x+y.$$
\end{lemma}
\begin{proof}
For the reader's convenience, we include a proof of this elementary lemma.
Denote by $2\e$ the minimal distance between the points and let $\psi:\R \to [0,1] $ be the bump function such that 
$$\psi(x)= \left\{ \begin{matrix}1 &{\rm \ for\ }& x \in [-\frac{\e}{2},\frac{\e}{2}], \\
  0 &{\rm \ for\ }& x \notin [-\e,\e] .\end{matrix} \right. $$
Denote by $A_\e, B_\e$ tubular neighbourhoods of size $\e$ and put $(A+B)_\e:=A_\e + B_\e$. The restriction of the sum map above the complement of $(A+B)_\e$ retracts by deformation on the complement over $A+B$.

For a complex number $z \in \C$, we use the subscript  $z_1$ for its real part and $z_2$ for its imaginary part. Denote the horizontal and vertical distance $d_j$ by $d_j(x,y)=x_j-y_j$, for $x,y\in \C$. Then, we put
$$a_j(x):=d_j(x,A),\ b_j:=d_j(x,B). $$
 Consider the vector fields
$$X_j=\frac{1-\psi(a_j(x))+\psi(b_j(x))}{2}\d_{x_j} + \frac{1+\psi(a_j(x))-\psi(b_j(x))}{2}\d_{y_j} .$$
Near $A$, we have $ X_j=  \d_{y_j} $, while near $B$, we get that $ X_j=  \d_{x_j}$. 
Away from these sets we have $ X_j=  \d_{x_j}+  \d_{y_j}$ and the vector field $X_j$ lifts
$\d_{\xi_j}  $. For $j,k\in \{ 1,2 \}$ and $j \neq k$, we have
$$\d_{x_j} a_k(x)=\d_{y_j} a_k(x)=\d_{x_j} b_k(x)=\d_{y_j} b_k(x)=0,$$
thus the vector fields $X_1,X_2$ commute and hence define a local trivialisation of the bundle. This proves the lemma and concludes the proof of the proposition.
\end{proof}
\end{proof}
\subsection*{Generalisation to higher dimensions}
The $\bullet$-product which appears in the $*$-product of $f,g \in \C\{\xi,\l\}$, with $\l=(q,p)$, and determined by \eqref{eq-convolpol}, is related to additive convolution by the formula:
\begin{equation}
f \bullet g (\xi,\l) = \int_0^{\xi} (\partial_\xi f)(\xi',\l) g(\xi-\xi',\l) d\xi'+f(0,\l)g(\xi,\l).\label{eq-convolbullet}
\end{equation}
Notice that the variables $q,p$ can be considered as a parameter $\l=(q,p)$. We may adapt Proposition \ref{P::closed} to this situation:

\begin{proposition} 
If $f,g \in \C\{\xi,\l \}$ are algebro-resurgent functions then the  product $f \bullet g$ is also algebro-resurgent.
 \end{proposition}
\begin{proof}
The set
$$ \S_f \bullet \S_g:=\S_f \cup \S_g \cup \{ (\l,x)+(\l,y): (\l,x) \in \S_f,\ (\l,y) \in \S_g \} $$
is an algebraic variety of $\C^{d+1}$, $d =\dim  \C\{ \xi,\l \}$. Thus, one can find a Zariski open subset $U \subset \C^d$ over which
the map
$$\C^{d+1} \setminus \left(  \S_f \bullet \S_g \right) \to \C^d, (\l,\xi) \mapsto \l $$
defines a locally trivial fibration. On each fibre of this fibration, one can repeat the proof of the one variable case. This proves the proposition.
\end{proof}

\subsection*{A closed formula for hypergeometric functions.}
Contrary to the  $\star$-product, the dual $*$-product defines a non-commutative algebra for {\em analytic} power series (Proposition~\ref{P::Borel}).  According to Theorem \ref{T::product}, algebro-resurgent power series form a $*$-subalgebra of functions having endless analytic continuation. The following proposition shows that algebraic functions do not form a $*$-subalgebra~:
\begin{proposition}
\label{ex-hyperg}
The hypergeometric function 
 $$F\left(-\a,-\b, 1; \xi\right):=\sum_{k=0}^{\infty}\begin{pmatrix} \alpha\\ k  \end{pmatrix}\begin{pmatrix}\beta\\ k  \end{pmatrix} \xi^k$$
  satisfies the identity
  \[(1+p)^{\a} * (1+q)^{\b}=(1+p)^{\a}(1+q)^{\b} F(-\a,-\b, 1;\frac{\xi}{(1+p)(1+q)}). \]
\end{proposition}
\begin{proof}
Although this is expected by the description of the singularities of the $*$-product, we prove the result directly by a naive direct computation.

First, note that if $f$ and $g$ do not depend on $\xi$, then the $\bullet$-product reduces to the
ordinary product: $f(q,p)\bullet g(q,p)=f(q,p)g(q,p)$. In the case of one degree of freedom ($n=1)$, we get
\begin{eqnarray*}
f * g (\xi,q,p)&=&\frac{1}{2\pi i} \int_{\gamma} f(q,y)g(x,p) Res\left(\frac{dx\w dy}{(x-q)(y-p)-\xi}\right)\\
&=&\frac{1}{2\pi i} \oint f(q,p+\frac{\xi}{x})g(q+x,p) \frac{dx}{x}.
\end{eqnarray*}
Hence we see that the product $f * g $ of two elements that do 
not depend on $\xi$  is just equal to a certain {\em Hadamard product}~\cite{Hadamard_produit}. Then,
\[(1+p)^{\a} * (1+q)^{\b}= \frac{1}{2\pi i}\oint (1+p+\frac{\xi}{x})^{\a}(1+q+x)^{\b} \frac{dx}{x}.\]
After expanding the powers with the binomial theorem we see that the integral picks out corresponding $x$-powers and we get
\[(1+p)^{\a} * (1+q)^{\b}=(1+p)^{\a}(1+q)^{\b}F(-\a,-\b, 1;\frac{\xi}{(1+p)(1+q)}).\]
This proves the proposition.
\end{proof}
In particular, there might exist many closed formul\ae relating modular functions to $*$-products. For instance, if we take $\a=\b=-1/2$ then
\begin{equation*}
\frac{1}{\sqrt{1+p}}*\frac{1}{\sqrt{1+q}}(\xi,0,0)=1+\left(\frac{1}{2}\right)^2\xi+\left(\frac{1.3}{2.4}\right)^2\xi^2+\left(\frac{1.3.5}{2.4.6}\right)^2\xi^3+\ldots
\end{equation*}
which is exactly the elliptic modular function $\frac{2}{\pi}K(k)$, for
$k=\sqrt{\xi}$, with
$$K(k)=\int_{0}^{1} \frac{dx}{\sqrt{(1-x^2)(1-k^2x^2)}}=\int_0^{\pi/2} \frac{d\varphi}{\sqrt{1-k^2\sin^2 \varphi}} .$$

\section{Outlook.}

Thanks to an integral formula proved in Proposition \ref{P::integral} for the $*$-product, we have been able to continue analytically this product. Indeed, we introduced the notion of algebro-resurgence in Definition \ref{D::ar} and we proved that the set $\cQ^A \subset \C\{\xi,q,p\}$ of all algebro-resurgent germs forms an algebra under the $*$-product in Theorem \ref{T::product}. Note that we obtain as a byproduct that the class of Gevrey series in $t$-variable which is Borel dual to the algebro-resurgent germs is an algebra for the $\star$-product, and it contains in particular the Euler series. Exponential small quantities - which encode interesting quantum effects - can be seen for example as the difference between the two Euler functions $E_\pm$ involved in \eqref{eq-smallexp}, which admit both the Euler series as asymptotic series. In a Borel dual point a view, these exponential small quantities should be interpreted in terms of the singularities of functions in $\cQ^A$, which is one of the main ideas in resurgence theory.

As a particular case of algebro-resurgence, any $*$-product of algebraic power series in
$\xi, q,p$ (that is, elements from the Henselian local ring) is algebro-resurgent. In fact, algebraic power series are even {\em holonomic}, meaning that they satisfy a holonomic system of differential equations and the subset $\cQ^H \subset \cQ^A$ of holonomic power series happens to be closed under the $*$-product. This is partly a consequence of the integral formula of Proposition \ref{P::integral}: the stability under integration follows from the fact that integrals of vanishing cycles always satisfy a Picard-Fuchs type equation. The stability under the convolution of functions satisfying a linear differential equation is a classical theorem of {\sc Hurwitz}. Details will appear elsewhere.

The algebra $\cQ^A$ we have constructed here seems to be rich enough to capture interesting quantum effects. One of the main difficulties in {\sc Pham}'s approach to {\sc Voros-Zinn-Justin} conjectures was indeed the absence of a convenient tool to describe singularities arising from algebraic operations. As we saw here, the singularities of a star-product can be explicitly described in the Borel plane. This led to the observation that starting from certain algebraic functions, one ends up with hypergeometric functions (see Proposition \ref{ex-hyperg}). On one side, this shows that the star-product immediately produces highly transcendental functions but from the point of view of singularities, they stay relatively simple since hypergeometric functions are just solutions to linear differential equations with a finite number of poles.

In this paper, we gave an abstract description of singularities and apply it only on the simple example of hypergeometric functions. However, we expect that for the case of the anharmonic oscillator, we should be able to obtain an explicit description of the singularities such as conjectured by {\sc Delabaere-Pham} \cite{delabaere1997unfolding} although these might involve complicated special functions.

There are also bigger algebras one could consider. Actually, the  $*$-exponential maps the algebra $\cQ^A$ to a bigger one, and there are 
several ways in which one could try to enlarge $\cQ^A$ so that the exponential maps the algebra to itself. {\sc Pham} and coworkers define a subspace $\cR \subset \C\{\xi\}$ of {\em resurgent germs} in one variable,
by saying that $f \in \cR$ if for all $L >0$ there is a finite set $\S_f(L)\subset \C$, such that all
$f$ can be analytically continued along all paths length $\le L$ that avoid $\S_f(L)$~\cite{Pham}. In \cite{Pham_livre_resurgence}, the authors sketch an argument that $\cR$ is closed under convolution, and it would be tempting to try to construct an analogue quantum $*$-algebra of this convolution algebra, using the integral formula of Proposition \ref{P::integral}. However, as observed by 
{\sc Delabaere}, {\sc Ou} and {\sc Sauzin}, there are some imprecisions in the original proofs of the convolution theorem for resurgent germs~\cite{Ou_resurgence,Sauzin_resurgence}.  We can see that the statement is true, by Proposition~\ref{P::closed}, if the singularity set is finite. Note that in the case the singularity set is a semi-group, detailed proofs have been given by {\sc Ou} for one-dimensional semi-groups and by {\sc Sauzin} in the two dimensional case.
\subsection*{Acknowledgments:} A. de Goursac acknowledges the F.R.S.-FNRS for financial support as well as the Belgian Interuniversity Attraction Pole (IAP) for support within the framework ``Dynamics, Geometry and Statistical Physics'' (DYGEST). M. Garay and A. de Goursac also acknowledge the Max Planck Institut f\"ur Mathematik in Bonn for invitation during this work.
\bibliographystyle{amsplain}
\bibliography{master}

\providecommand{\bysame}{\leavevmode\hbox to3em{\hrulefill}\thinspace}
\providecommand{\MR}{\relax\ifhmode\unskip\space\fi MR }
\providecommand{\MRhref}[2]{%
  \href{http://www.ams.org/mathscinet-getitem?mr=#1}{#2}
}
\providecommand{\href}[2]{#2}
\begin{thebibliography}{10}

\bibitem{Flato}
F.~Bayen, M.~Flato, C.~Fronsdal, A.~Lichnerowicz, and D.~Sternheimer,
  \emph{Deformation theory and quantization. I. Deformations of symplectic
  structures.}, Ann. Physics \textbf{111} (1978), no.~1, 61--110.

\bibitem{bieliavsky2011deformation}
P.~Bieliavsky and V.~Gayral, \emph{{Deformation Quantization for Actions of
  K\"ahlerian Lie Groups}}, preprint arXiv:1109.3419 (2013), To appear in Mem.
  Amer. Math. Soc.

\bibitem{birkhoff1933quantum}
G.D. Birkhoff, \emph{Quantum mechanics and asymptotic series}, Bulletin of the
  American Mathematical Society \textbf{39} (1933), no.~10, 681--700.

\bibitem{borel1901leccons}
{\'E}.~Borel, \emph{Le{\c{c}}ons sur les s{\'e}ries divergentes}, vol.~2,
  Gauthier-Villars, 1901.

\bibitem{Born_Jordan}
M.~Born and P.~Jordan, \emph{{Zur Quantenmechanik}}, Zeit. f\" ur Phys.
  \textbf{34} (1925), 858--888.

\bibitem{Boutet_Kree}
L.~Boutet~de Monvel and P.~Kr\'ee, \emph{{Pseudo differential operators and
  Gevrey classes}}, Annales de l'Institut Fourier \textbf{17} (1967), no.~1,
  295--323.

\bibitem{Pham_livre_resurgence}
B.~Candelpergher, J.C. Nosmas, and F.~Pham, \emph{{Approche de la
  r\'esurgence}}, Hermann, 1993, 289 pp.

\bibitem{Chabat_Introduction_II}
Boris Chabat, \emph{Introduction \`a l'analyse complexe, tome II}, Mir, 1976.

\bibitem{delabaere1997unfolding}
E.~Delabaere and F.~Pham, \emph{Unfolding the quartic oscillator}, Annals of
  Physics \textbf{261} (1997), no.~2, 180--218.

\bibitem{Sjostrand}
M.~Dimassi and J.~Sj\"ostrand, \emph{Spectral asymptotics in the semi-classical
  limit}, Lecture Note Series, no. 268, Cambridge University Press, 1999.

\bibitem{Ecalle_fonctions}
J.~\'Ecalle, \emph{Les fonctions r\'esurgentes, vol. 1, alg\`ebres de fonctions
  r\'esurgentes}, Pub. Math. Orsay (1981).

\bibitem{eremenko2004geometric}
A.~Eremenko, \emph{Geometric theory of meromorphic functions}, Contemporary
  Mathematics \textbf{355} (2004), 221--230.

\bibitem{Eremenko_Gabrielov_quartic}
A.~Eremenko and A.~Gabrielov, \emph{{Analytic continuation of eigenvalues of a
  quartic oscillator}}, Comm. Math. Physics \textbf{287} (2009), 431--457.

\bibitem{Euler_divergent}
L.~Euler, \emph{De seriebus divergentibus}, Novi Commentarii academiae
  scientiarum Petropolitanae \textbf{5} (1760), 205--237, German translation
  available on http://www.eulerarchive.org/, E247.

\bibitem{quantique}
M.D. Garay, \emph{Perturbative expansions in quantum mechanics}, Annales de
  l'institut Fourier \textbf{59} (2009), no.~5, 2061--2101.

\bibitem{gross1918ganze}
W.~Gross, \emph{{Eine ganze Funktion, f{\"u}r die jede komplexe Zahl
  Konvergenzwert ist}}, Mathematische Annalen \textbf{79} (1918), no.~1,
  201--208.

\bibitem{Hadamard_produit}
J.~Hadamard, \emph{Th{\'e}or{\`e}me sur les s{\'e}ries enti{\`e}res}, Acta
  Mathematica \textbf{22} (1899), no.~1, 55--63.

\bibitem{Hurwitz_produits}
A.~Hurwitz, \emph{{Sur une th\'eor\`eme de M. Hadamard}}, Comptes Rendus \`a
  l'Acad\'emie des Sciences \textbf{128} (1899), 350--353.

\bibitem{Moyal}
J.~E. Moyal, \emph{Quantum mechanics as a statistical theory}, Proc. Cambridge
  Philos. Soc. \textbf{45} (1949), 99--124.

\bibitem{Ou_resurgence}
Y.~Ou, \emph{{On the stability by convolution product of a resurgent algebra}},
  Annales de la facult\'e des sciences de Toulouse Math\'ematiques \textbf{19}
  (2010), no.~3-4, 687--705.

\bibitem{Pham_resurgence}
F.~Pham, \emph{Resurgence, quantized canonical transformations, and
  multi-instanton expansions}, Algebraic analysis (M.~Kashiwara and T.~Kawai,
  eds.), vol.~II, Academic Press, Boston, MA, 1988, Papers dedicated to
  Professor Mikio Sato on the occasion of his sixtieth birthday, pp.~699--726.

\bibitem{Pham}
\bysame, \emph{{Multiple turning points in exact WKB analysis (variations on a
  theme of Stokes)}}, Towards the exact {WKB} analysis of differential
  equations, linear or non linear (C.~Howls, T.~Kawai, and Y.~Takei, eds.),
  Kyoto University Press, 2000, pp.~71--85.

\bibitem{Reed_Simon}
M.~Reed and B.~Simon, \emph{Methods of modern mathematical physics}, vol.~IV,
  Academic Press, 1978.

\bibitem{rieffel1989deformation}
M.~A. Rieffel, \emph{Deformation Quantization for actions of R(D)},
  Memoirs of the American Mathematical Society \textbf{106} (1993), R6.

\bibitem{Sauzin_resurgence}
D.~Sauzin, \emph{{On the stability under convolution of resurgent functions}},
  arXiv:1201.1073, 2012.

\bibitem{Sjostrand_asterisque}
J.~Sj\"ostrand, \emph{Singularit\'es analytiques microlocales}, Ast\'erisque
  \textbf{95} (1982), 1--166.

\bibitem{Verdier_w}
J.-L. Verdier, \emph{Stratifications de Whitney et th\'eor\`eme de
  Bertini-Sard}, Inventiones Mathematicae \textbf{36} (1976), 295--312.

\bibitem{Voros}
A.~Voros, \emph{{ The return of the quartic oscillator (the complex WKB
  method)}}, Ann. Inst. H. Poincar\'e, Phys. Th\'eo. \textbf{39} (1983),
  211--338.

\bibitem{Zinn_Justin_nuclear}
J.~Zinn-Justin, \emph{Multi-instanton contributions in quantum mechanics, 2},
  Nucl.Phys.B \textbf{218} (1983), 333--348.

\bibitem{zinn2004multi}
J.~Zinn-Justin and U.~D. Jentschura, \emph{{Multi-instantons and exact results
  I: conjectures, WKB expansions, and instanton interactions}}, Annals of
  Physics \textbf{313} (2004), no.~1, 197--267.

\bibitem{zworski2012semiclassical}
M.~Zworski, \emph{Semiclassical analysis}, vol. 138, AMS Bookstore, 2012.

\end{thebibliography}
\vspace*{0.2cm}

\noindent {\sc Mauricio Garay}\\
{\footnotesize On leave from \\
Institut f\"ur Mathematik, FB 08 Physik, Mathematik und Informatik,\\
Johannes Gutenberg-Universit\"at\\
55099 Mainz, Germany\\
garay91@gmail.com}
\vspace*{0.2cm}

\noindent{\sc Axel de Goursac}\\
{\footnotesize Charg\'e de Recherche au F.R.S.-FNRS,\\ IRMP, Universit\'e Catholique de Louvain,\\ Chemin du Cyclotron, 2, B-1348 Louvain-la-Neuve, Belgium\\
Axelmg@melix.net}
\vspace*{0.2cm}

\noindent{\sc Duco van Straten}\\
{\footnotesize Institut f\"ur Mathematik, FB 08 Physik, Mathematik und Informatik,\\
Johannes Gutenberg-Universit\"at\\
55099 Mainz, Germany\\
straten@mathematik.uni-mainz.de}
\end{document}